\documentclass[12pt]{article}
\usepackage{color}
\usepackage{amsmath,amssymb}
\usepackage{amsthm}
\usepackage{makeidx,graphicx}
\def\ind{{\mathchoice {\rm 1\mskip-4mu l} {\rm 1\mskip-4mu l}
{\rm 1\mskip-4.5mu l} {\rm 1\mskip-5mu l}}}

\title{ Non-Asymptotic Mean-Field Games}
\author{ Hamidou Tembine (IEEE Senior Member)
\thanks{We are grateful to many seminar and conference participants such as those at  NETGCOOP 2012, at the Second Workshop on Mean Field Games (Padova, Italy, September 2013) for their valuable comments and suggestions on the preliminary versions of this work. }
\thanks{H. Tembine is a member of the KAUST Strategic Research Initiative Center for Uncertainty Quantification in Computational Science and Engineering.
        {\tt\small E-mail: tembine@ieee.org}}
        }
 \newtheorem{example}{Example}
 \newtheorem{thm}{Result}
\newtheorem{defi}{Definition}

\newtheorem{rem}{Remark}

\begin{document}

\maketitle

\begin{abstract}
Mean-field games have been studied under the assumption of very large number of players. For such large systems, the basic idea consists to approximate  large games by a stylized game model with a continuum of players. The approach has been shown to be useful in some applications.  However, the stylized game model with  continuum of decision-makers is rarely observed in practice and the approximation proposed in the asymptotic  regime is meaningless for networks with few entities.
In this paper we propose a mean-field framework that is suitable not only for large systems but also for a small world with few number of entities. The applicability of the proposed framework is illustrated through various examples including  dynamic auction with asymmetric valuation distributions, and spiteful bidders. 
\end{abstract}

\section{Introduction} 

Recently there have been renewed interests in large-scale interaction in several research disciplines, with its uses in wireless networks, big data, cyber-physical systems, financial markets, intelligent transportation systems, smart grid, crowd safety, social cloud networks and smarter cities.

 In mathematical physics, most of models are analyzed in the asymptotic regime when the size of the system grows without bounds. As an example, the McKean-Vlasov model~\cite{kac1,kac2,villani} for interacting particles is analyzed  when the number of particles tends to infinity. Such an approach is referred to as mean field approach. The seminal works of Sznitman~\cite{dyna1} in the 1980s and the more recent work of Kotolenez \& Kurtz \cite{dyna2} show that the asymptotic system provides a good approximation of the finite system in the following sense: For any tolerance level $\epsilon>0$ there exists a population size $n_{\epsilon}$ such that for any $n\geq n_{\epsilon},$ the error gap between the solution of the infinite system and the system with size $n$ is at most $\epsilon.$ Moreover, the work in \cite{dyna2} shows that the number $n_{\epsilon}$ is in order of $O\left( \log(\frac{1}{\epsilon}) \frac{1}{\epsilon^{d+2}}\right)$ for a class of smooth functions, where $d$ denotes the dimension of the space.  Thus, for $n< n_{\epsilon}$ this current theory  does not give an approximation that is meaningful.

In queueing theory, the number of customers is usually assumed to be large or follows a certain distribution with unbounded support (e.g., exponential, Poisson etc) and the buffer size (queue) can be infinite. However, many applications of interests such as airport boarding queues, supermarket queues, restaurant queue, iphone/ipad waiting queue involve a finite number of customers/travelers. Approximation by a continuum of decision-makers may not reflect the reality.  For example the number of clients in the supermarket queue cannot exceed the size of available capacity of markets and there is a certain distance between the clients to be respected. In other words, human behaviors are not necessarily like  standard fluid dynamics.
In game theory,
 the rapidly emerging field of mean-field games~\cite{lions} is addressing  behavioral and algorithmic issues \cite{learningbook} for mathematical models with continuum of players. We refer the reader to \cite{lionssurvey} for a survey on (asymptotic) mean field games.

The classical works mentioned above provide rich mathematical foundations and equilibrium concepts in the asymptotic regime, but relatively little in the way of computational and representational insights that would allow for few number of players. Most of the  mean-field game models consider a continuum of players, which seems not realistic in terms of most applications of interests. 
Below we give some limitations of the {\it asymptotic} mean-field approaches in engineering and in economics:
\begin{itemize}\item
In wireless networks, the number of interacting nodes at the same slot in the same range is finite and currently the capacity/bandwidth of the system is limited. Therefore, a mean-field model for infinite capacity and infinite number of nodes is not plausible. The result of infinite system may not capture the real system with only few number of nodes. 
\item In most of the current markets, the number of traders is {\it finite}. In that context it is well known that the Bayesian-Cournot game may not have an ex-post equilibrium whenever the number of traders is finite. However, the infinite game with continuum of traders has a pure (static mean-field) equilibrium. If our prediction is the "mean-field equilibrium", in what sense the (static) mean-field ex-post equilibrium \cite{mfsg} captures the finite system?
\end{itemize}

Our primarily goal in this article is to provide a simple and easy to check condition such that mean-field theory can be used for finite-scale which we call {\it non-asymptotic mean-field} approach. We investigate the nonasymptotic mean-field under two basic conditions. The first condition is indistinguishability (or interchangeability) of the payoff functions. The indistinguishability
property is easy to verify. The indistinguishability assumption is implicitly used in the classical (static) mean-field analysis including the seminal works of Aumann 1964~\cite{aumann},  Selten 1970~\cite{selten}, Schmeidler 1973. This assumption is also implicitly used in the dynamic version of mean-field games by Jovanovic \& Rosenthal 1988\cite{ros1}, Benamou \& Brenier 2000~\cite{benamoubrenier}  and Lasry \& Lions 2007~\cite{lions}.
The second condition is the (regularity) smoothness of the payoff functions.
The regularity property
is relatively easy to check.

Based on these two conditions, we present a simple approximation framework for finite horizon mean-field systems. The framework can be easily extended to infinite horizon case.
The non-asymptotic mean field  approach is based on a simple observation that  the
many effects of different actions cancel out when the payoff
is indistinguishable. Nevertheless, it can lead to a significant
simplification of mathematical mean-field models in finite regime. The approach presented here is non-asymptotic and is unrelated to the mean-field convergence that
originates from  law of large numbers (and its generalization to de Finetti-Hewitt-Savage functional mean-field convergence) in large populations. The non-asymptotic mean field approach
 holds even when there are only few players in a game,
or  few nodes in a network. 

The idea presented here is inspired from the works in \cite{dynamict1,reftv2,pnas} on the so-called averaging principle. These previous works are limited to static and one-shot games.  Here we use that idea not only for static games but also for  dynamic mean-field games. One of the motivations of the asymptotic mean field game approach is that it may reduce the complexity analysis of large systems. The present work  goes beyond that. We believe that if the complexity of the infinite system can be reduced easily then, the finite system can also be studied using a non-asymptotic mean-field approach.

In order to apply the mean-field approach to a system with arbitrary number of players, we shall exploit more the structure of objective function and the main assumption of the model which is the indistinguishability property, i.e., the performance index is unchanged if one permutes the label of the players. This is what we will do in this work. The aggregative  structure of the problem and the indistinguishability property of the players are used to derive an error bound for any number of players. Interestingly, our result holds not only  for large number of players but also for few number of players. For  example, for $n=5$ players, there is no systematic way to apply the theory developed in the previous works~\cite{lions,aumann} but the non-asymptotic mean-field result presented here could be applied.
The non-asymptotic mean-field result  does not impose additional assumptions on the payoff function. We show that the indistinguishability property  provides an accurate error bound for any system size.  We show that the total equilibrium payoff with heterogeneous parameters can be approximated by the symmetric payoff where the symmetry is the respect to the mean of those parameters. These parameters can be a real number, vector, matrix or a infinite functional. The proof of the approximation error is essentially based on a Taylor expansion which cancels out the first order terms due to indistinguishability property.

 We provide various examples where non-asymptotic mean-field interaction is required and the indistinguishability property could be exploited more efficiently. We present of queueing system with only few servers where closed-form expression of the waiting time is not available and the use of the present framework gives  appropriate bounds. As second main example focuses on dynamic auctions with asymmetric bidders that can be self-interested, malicious or  spiteful.
In models of first-price auctions, when bidders are ex ante heterogeneous, deriving explicit equilibrium bid functions is an open issue. Due to the boundary-value problem nature of the equilibrium,  numerical methods remain challenging issue.
Recent theoretical research concerning asymmetric auctions have determined some qualitative properties
these bid functions must satisfy when certain conditions are met. Here we propose an accurate approximation based on non-asymptotic mean field game approach and examine the relative
expected payoffs of bidders and the seller revenue (which is indistinguishable) to decide whether the approximate solutions are
consistent with theory.

The remainder of the paper is structured as follows. In Section \ref{staticscale} we present a mean field system with arbitrary number of interacting entities and propose a nonasymptotic static mean field framework. In Section \ref{secdynamicscale} we extend our basic results in a dynamic setup. In Section \ref{basicapplication} we present  applications of nonasymptotic mean-field approach to collaborative effort game, approximation of queueing delay performance and  computation of error bound of equilibrium bids in dynamic auction with asymmetric bidders.

We summarize some of the notations in Table \ref{tablenotationjournal}.
\begin{table}[htb]
\caption{Summary of Notations} \label{tablenotationjournal}
\begin{center}
\begin{tabular}{ll}
\hline
  Symbol & Meaning \\ \hline
  $\mathcal{N}$ & set of potential minor players \\
  $n$ & cardinality of $\mathcal{N}$\\
  $\mathcal{A}$ & action space\\
  $a_{j}$ & action of player $j$  \\
  $r_g(a_1,\ldots,a_n)$ & global payoff function of the major player\\
  $\ind_{\{.\}}$ & indicator function.\\
  $\bar{m}^{\bigotimes n}$ & $=(\bar{m},\ldots,\bar{m})$\\
  $\tau_j$ & strategy of player $j$\\
  $R_{j,T}$ & long-term payoff of player $j$ with horizon $[0,T-1]$\\
  \\ \hline
\end{tabular}
\end{center}
\end{table}

\section{Mean-field system for arbitrary number of entities} \label{staticscale}

{ 
Consider an interactive system with $n+1\geq 2$ entities (players) consisting of $n\geq 1 $ generic minor players and one major player (called designer).  The major player has a binary decision set. It consists to propose ($P$) or not to propose ($\bar{P}$) a secondary game between the minors,  thus its action set is $\{P, \bar{P} \}. $  Each of the minor players has to make a decision. Each decision variable  $a_j$ of a minor player $j\in \mathcal{N}:=\{1,2,\ldots, n\}$ belongs to a Polish\footnote{A Polish space $E$ is a Separable topological space $E$ for which there exists a compatible metric $d$ such that $(E,d)$ is a complete metric space. Here, "`separable"' means has a countable dense subset.} space $\mathcal{A}.$
Each minor  player has a payoff function $\hat{r}_{j}(a_0, a_1,\ldots,a_n)\in \mathbb{R}$ where $a_0\in \{P, \bar{P} \}$ is the action of the major player. The major player has its own payoff function that could be the global  performance  of the  minor players or another generic payoff. The payoff function of the major player is captured by a certain function
$\hat{r}_{g}(a_0, a_1,\ldots,a_n)\in \mathbb{R}$ which we call global payoff function. 

The collection $$
\left\{\{0,1,\ldots, n\}, \{P, \bar{P} \}, \mathcal{A}, \hat{r}_{g}, \hat{r}_{j}\right\},
$$
defines a game in strategic-form (or normal form).

 Since the action set of the major player reduces to a binary set, the decision will be driven by the comparison between $\hat{r}_{g}(P, a_1,\ldots,a_n) $  and $r_{g}(\bar{P}, a_1,\ldots,a_n).$ The payoff $\hat{r}_{g}(P, a_1,\ldots,a_n) $ when the continuation game is not proposed can be fixed a certain constant. Therefore we focus on the analysis of 
the payoff $\hat{r}_{g}(P, a_1,\ldots,a_n) $  which we denote  by $r_{g}(a_1,\ldots,a_n).$ Similarly, $\hat{r}_{j}(P, a_1,\ldots,a_n) $ is denoted simply by ${r}_{j}( a_1,\ldots,a_n).$

}

\subsection{Main Assumptions on the structure of payoff function}

{\bf Assumption A0:} Indistinguishability.  We assume that the global payoff function is invariant by permuting the index of the  minor players, i.e., $$r_g(a_1,\ldots,a_i,\ldots,a_j,\ldots,a_n)=$$ $$r_g(a_{\pi(1)},\ldots,a_{\pi(i)},\ldots,a_{\pi(j)},\ldots,a_{\pi(n)}),
$$
for every permutation $\pi:\ \mathcal{N} \longrightarrow  \mathcal{N}, $  where  $\mathcal{N}:=\{1,2,\ldots,n\}.$

 To verify $A0$, it suffices to check for pairwise interchangeability, i.e.,   permutation of any two of the coordinates. In mathematics, the indistinguishability property is sometimes referred as  symmetric function, i.e., one whose value at any n-tuple of arguments is the same as its value at any permutation of that $n-$tuple.

{\bf Assumption A1:} Smoothness. We assume that the objective function $r_g$ is (locally) twice differentiable with the respect to the variables.

It is important to notice that the assumption A0 can be easily checked by designers, engineers and non-specialists.
In practice, A0 will result in functions that can be expressed in terms of the mean $\bar{m}=\frac{1}{n}\sum_{j=1}^n a_j,$ or other aggregative terms such as  $ \frac{1}{n}\sum_{j=1}^n a_j^2,\ \  \ \frac{1}{n}\sum_{j=1}^n \phi(a_j,\bar{m}) \ldots,$ $\left(\prod_{j=1}^n a_j\right),$ etc. Assumption A0 is implicitly used in \cite{aumann,selten,ros1,lions}.

Our goal is provide a useful approximation and error bound for the global  payoff $r_g$ in an equilibrium or in function of the parameters of the game.

\subsection{Applicability of the payoff structure}
 These type of payoff functions  have wide range of applications:

 In economics and financial markets, the market price (of products, good, phones, laptops, etc) is influenced by the total demand and total supply. Examples include
   \begin{itemize}
   \item Public good\footnote{Goods are called public if one person's consumption of them does not preclude consumption by others. Typical examples are television programs and uncongested roads.
} provisioning with total payoff $g(\bar{m}).$
   \item Beauty contest games with payoff $R- \parallel a_j-p \bar{m}\parallel, \ 0<p<1$
   \item Cournot oligopoly model with payoff $a_jp(\frac{\sum_{j'=1}^n a_{j'}}{n})-c(a_j).$
   \end{itemize}

  In queueing theory, the task completion of data centers or  servers is influenced by the mean of how much the other data centers/servers can serve.
 In resource sharing problems, the utility/disutility of a player depends on the demand of the other players. Examples include cost sharing in coalitional system and  capacity and bandwidth sharing in cloud networking.
 In wireless networks, the performance of a  wireless node is influenced by the interference created by the other  transmitters.
 In congestion control, the delay of a network depends on the aggregate (total) flow and the congestion level of the links/routes.

\subsection{Approximation for static games}

Next we provide the basic  results that hold for both  non-asymptotic and asymptotic static systems. 

\begin{thm} Assume that $A0$ and $A1$ hold. Then, the following results hold:
\begin{itemize}
\item $\partial_{a_j}r_g(\bar{m}^{\bigotimes n})=\partial_{a_1}r_g(\bar{m}^{\bigotimes n})$ where $\bar{m}^{\bigotimes n}:=(\bar{m},\ldots,\bar{m})$ and
    \begin{eqnarray} \label{meanaction}
    \bar{m}=\frac{1}{n}\sum_{j=1}^n a_j=\int_{b\in \mathcal{A}} b \left[ \frac{1}{n}\sum_{j=1}^n \delta_{a_j}\right](db),
    \end{eqnarray} $\delta_{a_j}$ is the Dirac measure concentrated at the point $a_j,$
     \begin{eqnarray} \label{meanfield2}m=\frac{1}{n}\sum_{j=1}^n \delta_{a_j}
      \end{eqnarray}
\item
The structure of the payoff function implies that the first order term in the Taylor expansion is cancelled out.
\item The cross-derivatives are independent of the labels: $\partial^2_{a_ia_j}r_g(\bar{m}^{\bigotimes n})=\partial^2_{a_1a_2}r_g(\bar{m}^{\bigotimes n})$
\end{itemize}
\end{thm}

Note that this theorem can be used for games with continuous action space as well as for games with discrete action space via mixed extensions. Examples of games that satisfy A0-A1 includes Prisoner Dilemma, Battle of Sex, Hawk-Dove, coordination games, anti-coordination games, minority games, matching pennies, etc. 

\begin{proof}
The first item is immediately proved by using the indistinguishability property in the definition of directional derivative.
For the second item we use the first item and the relation (\ref{meanaction}). The first order derivative term is
\begin{align}
&\sum_{j=1}^n (a_j-\bar{m})\partial_{a_j}r_g(a)\\ \nonumber &= \sum_{j=1}^n( a_j-\bar{m})\partial_{a_1}r_g(a)\\ \nonumber
&= \partial_{a_1}r_g(a) \left[ \sum_{j=1}^n (a_j-\bar{m} )\right]\\ \nonumber
&=  \partial_{a_1}r_g(a) \left[\left(\sum_{j=1}^n a_j\right) -n\bar{m}\right]\\ 
&= \partial_{a_1}r_g(a) \left[\sum_{j=1}^n a_j-\sum_{j=1}^n a_j\right]=0.
\end{align}
The third item uses another exchange of positions between $i$ and $j.$
\end{proof}

\begin{thm} \label{approxt1}
Suppose that the payoff function $r_g$ satisfies the assumptions $A0$ and $A1.$  Assume that $a$ is in a small neighborhood of the mean vector $\bar{m}^{\bigotimes n}:=(\bar{m},\ldots,\bar{m})$ i.e., there is a  small positive number $c_{\bar{m},r_g}$ which may depend on $\bar{m}$ and the function $r_g$ such that $ \parallel (a_1, \ldots,a_n)-\bar{m}^{\bigotimes n}\parallel \leq c_{\bar{m},r_g}$ then
$
 r_g(a)-\bar{r}(\bar{m})=O\left(\parallel a-\bar{m}^{\bigotimes n} \parallel^2_2\right),$
and
$$
\parallel r_g(a)-\bar{r}(\bar{m})\parallel \leq \delta c_{\bar{m},r_g}^2
$$
where $\bar{r}(\bar{m}):=r_g(\bar{m},\ldots,\bar{m})=r_g(\bar{m}^{\bigotimes n}), \delta>0.$

\end{thm}

The proof of result \ref{approxt1} follows from the following result \ref{thmmain} which gives the explicit error bound:
\begin{thm}\label{thmmain}
Assume that $A0-A1$ hold. Then, the explicit error bound for arbitrary number of players is $r_g(a)-\bar{r}(\bar{m})$ is in order of $$\delta_{\bar{m},\bar{r}} \sum_{j=1}^n (a_j-\bar{m})^2,$$ where $$\delta_{\bar{m},\bar{r}}=|\frac{n}{2(n-1)}\left( -\frac{1}{n^2}\bar{r}''(\bar{m})+\partial^2_{a_1a_1}r_g(\bar{m}^{\bigotimes n})\right)| $$
\end{thm}
\begin{proof}
Let $O_2$ be the second order error term.
\begin{eqnarray}\nonumber
2O_2&=&\sum_{j,j'=1}^n (a_j-\bar{m})(a_{j'}-\bar{m})\partial^2_{a_ja_{j'}}r(\bar{m}^{\bigotimes n})\\ \nonumber
&=& \sum_{j,j', \ j\neq j'}^n (a_j-\bar{m})(a_{j'}-\bar{m})\underbrace{\partial^2_{a_ja_{j'}}r(\bar{m}^{\bigotimes n})}_{=\partial^2_{a_1a_2}r} \nonumber
\\ && + \sum_{k=1}^n(a_{k}-\bar{m})^2\underbrace{\partial^2_{a_ka_{k}}r}_{=\partial^2_{a_1a_1}r}\\ \nonumber
&=& (-\partial^2_{a_1a_2}r+\partial^2_{a_1a_{1}}r)\sum_{k=1}^n(a_{k}-\bar{m})^2
\end{eqnarray}
On the other hand, one has
\begin{eqnarray}
r(\bar{m}^{\bigotimes n}+\epsilon \ind_{})&=&\bar{r}(\bar{m}+\epsilon)
\end{eqnarray}
This implies that
\begin{eqnarray}
2 O_2&=&\epsilon^2 \bar{r}''\\ \nonumber
&=& n(n-1)\epsilon^2 (\partial^2_{a_1a_2}r)+ n\epsilon^2 (\partial^2_{a_1a_1}r)
\end{eqnarray}

Hence, $(\partial^2_{a_1a_2}r)= \frac{\bar{r}''(\bar{m})-n \partial^2_{a_1a_1}r}{n(n-1)}.$

\begin{eqnarray}\nonumber
O_2 &=& \frac{1}{2}\left[ \frac{-\bar{r}''(\bar{m})+n \partial^2_{a_1a_1}r}{n(n-1)} +\partial^2_{a_1a_1}r\right] \sum_{j=1}^n(a_{j}-\bar{m})^2\\ \nonumber
&=& \frac{1}{2}\left[ \frac{-\bar{r}''(\bar{m})+n^2 \partial^2_{a_1a_1}r}{n(n-1)}\right] \sum_{j=1}^n(a_{j}-\bar{m})^2\\ \nonumber
&=& \frac{n}{2(n-1)} \left[ -\frac{1}{n^2}\bar{r}''(\bar{m})+ \partial^2_{a_1a_1}r\right] \sum_{j=1}^n(a_{j}-\bar{m})^2,
\end{eqnarray}
 and we get the exact error for arbitrary system size.
 \end{proof}

\begin{rem}
In order to compute the error bound, one needs only $\bar{r}$, and $\partial^2_{a_1a_1}r_g(\bar{m}^{\bigotimes n}).$ The expression of the function $r_g(a)$ is not required for vector with non-symmetric components. This allows us to provide an approximation result for unknown payoff function as illustrated in Subsection \ref{queue example}.

\end{rem}
\begin{rem}
If $\bar{r}''$ is bounded by $\beta$  and  $\frac{\sum_{j=1}^n (a_j-\bar{m})^2}{n}\leq \sigma$ then
\begin{align}r_g(a)-\bar{r}(\bar{m})&\leq \frac{\beta}{2(n-1)} \frac{\sum_{j=1}^n (a_j-\bar{m})^2}{n}\\
&\leq \frac{\beta \sigma}{2(n-1)}
\leq \frac{\beta \epsilon^2}{2(n-1)} .
\end{align}
In particular, if the finite regime has a solution in a certain sense, that is $\epsilon-$close to a vector with symmetric component then, the non-asymptotic mean-field approach provides automatically an $O(\frac{ \epsilon^2}{2(n-1)})-$solution for any number of players $n\geq 2.$ This is a non-asymptotic result in the sense that it holds for all range of system size $n\geq 2.$ Also, by choosing  $\epsilon=\frac{1}{n^{\alpha}},\ \alpha\geq 1$ one gets an error bound in order of $\frac{1}{2 n^{2\alpha+1}}.$ Note that $\epsilon$ can be very small even if $n$ is not large. For example, with $n=2$ players and $\alpha=10$, one gets an error bound in order of $\frac{1}{2^{22}}$ which is satisfactory in terms of computational accuracy.
\end{rem}

\section{Dynamic setup}\label{secdynamicscale}
In this section we provide very useful approximation results for dynamic interactive systems \cite{marriage}. We consider a finite horizon with length $T\geq 1.$
\subsection{Non-asymptotic mean-field optimization}
Consider a major player who controls also the action to be dictated to $n$ minor entities. Assume that the major player aims to achieve a certain goal with objective function given by
the choice variables that the major player dictates to the minor entities. The objective function is
$R_{g,T}(\tau)=\sum_{t=0}^{T-1} r_g(a_t)$ where  $r_g: \ \mathcal{A}^n \rightarrow \mathcal{A}^n$ satisfies assumptions $A0-1$ and $a_t=(a_{1,t},\ldots,a_{n,t})\in \mathcal{A}^n$ is the choice variable at time $t.$  Let $\bar{m}_t=\frac{1}{n}\sum_{j=1}^n a_{j,t}$ be the  sequence of mean actions and set
$ \bar{r}_T(\bar{m})=\sum_{t=0}^{T-1}  \bar{r}(\bar{m}_t).$
\begin{thm} 
An explicit error bound for  $R_{g,T}(\tau)-\bar{r}_T(\bar{m})$ with arbitrary number of minor entities is given  by $$\delta_{T,\bar{m},\bar{r}} \parallel a-\bar{m}\parallel_{l^{2}_T}^2,$$ where
$$\delta_{T,\bar{m},\bar{r}}=\sup_{t\in \mathcal{T}}
|\frac{n}{2(n-1)}\left( -\frac{1}{n^2}\bar{r}''(\bar{m}_t)+
\partial^2_{a_1a_1}r_g(\bar{m}^{\bigotimes n}_t)\right)|
$$
and $l^{2}_T=\{(x_t)_{t\leq T-1} \ | \ \sum_{t=0}^{T-1} |x_t|^2 < +\infty \},$ and $\| a-\bar{m}\|_{l^{2}_T}=\sum_{t=0}^{T-1} |a_t-\bar{m}_t|^2  $
\end{thm}
See Result \ref{thmmain2} below for a proof.
\subsection{Non-asymptotic mean-field stochastic games}
Consider a stochastic game~\cite{shapley,gamenets} with $n$ minor players and one major player (designer). Time is discrete. Time space is $\mathcal{T}=\{0,1,\ldots,T-1\}$ where $T\geq 1.$ Each player $j$  has its individual state $s_{j,t},\ t\in\mathcal{T}$ which evolves according to a Markovian process. The action space a player depends on its current space, $\mathcal{A}(s_{j,t})\subset \mathbb{R}.$ A  pure strategy of  player $j$ at time $t$ is a mapping from history $\mathcal{H}_{j,t}$ up to $t$ to the current action space $\mathcal{A}(s_{j,t}).$  Denote by $\tau_j=(\tau_{j,t})_t$ the strategy of player $j.$  For $h_{j,t}\in \mathcal{H}_{j,t},\  \tau_{j,t}(h_{j,t})=a_{j,t}\in \mathcal{A}(s_{j,t}).$ The instantaneous payoff function of player $j$ is $r(s_{j,t},a_{j,t},m_t)$ where $m_t$ is the  state-action distribution, which satisfies the indistinguishability property with the respect to the other players. We assume that  the payoff function $r$ is smooth. 

The long-term payoff of player $j$ is
$$R_{j,T}(s_0,\tau)=\mathbb{E}\left[ \bar{g}(s_T)+\sum_{t\in \mathcal{T}}r(s_{j,t},a_{j,t},m_t)\ | \ (s_{j,0}, \ \tau_j)_{j\in\mathcal{N}}\right]
$$ where  $\bar{g}$ is the terminal payoff.
A strategy profile $\tau^*$ is a (Nash) equilibrium if no player can improve her payoff by unilateral deviation, i.e., for every player $j,$
{ $$R_{j,T}(s_0,\tau^*)=\max_{\tau_j}R_{j,T}(s_0,\tau_j,\tau_{-j}^*)$$
}

Let $V_{j,T}(s_0,\tau_{-j})$ be the value function of the bidder $j$ , i.e., it is the supremum, over all
possible bidding strategies, of the expectation of the payoff $R_{j,T}$ starting from
an initial state $s_0$ when the other players strategy profile is $\tau_{-j}:$ 

$$V_{j,T}(s_0,\tau_{-j})=\max_{\tau_j} [R_{j,T}(s_0,\tau) \ | \  \tau_{-j}, s_0].$$

Taking the expectation over the other players state, the recursive Bellman-Kolmogorov equation is given by
$$\begin{array}{c}
V_{j,t}(s_{j})=\sup_{a_j}\left[ r(s_j,a_j,m_t)+\mathbb{E}_{s'_j} V_{j,t+1}(s'_j\ | s_j,a_j,\tau) \right],\\
m_{t+1}\sim \mathbb{P}(.| \ m_t,\tau_t)\\
s_{j,t+1}\sim q(.|s_{j,t},a_{j,t},m_t,\tau_t)=\mathbb{P}(.|\ s_{j,t},a_{j,t},m_t,\tau_t)
\end{array}
$$
where $q$ and $\mathbb{P}$ define the transition probabilities between the states.
\begin{thm}
\label{thmmain2}
Let $a_{j,t}(h_{j,t})=\bar{m}_t(h_{j,t})+\epsilon \gamma_{j,t}(h_{j,t}),$ where
$$ \gamma_{j,t}(h_{j,t})=\frac{a_{j,t}(h_{j,t})-\bar{m}_t(h_{j,t} )}{\epsilon},$$
$$\bar{m}_t(h_{j,t})=\frac{1}{n}\sum_{j'=1}^n a_{j',t}(h_{j',t}),$$ and
$$\epsilon=\max_j \sup_{t\in \mathcal{T}}\sup_{h_{j,t}\in \mathcal{H}_{j,t}}\  \| a_{j,t}(h_{j,t})-\bar{m}_t(h_{j,t} ) \|,$$
Assume the state transition $q$ is continuous.
Then the total term payoff $R_g(s_0;\tau_1,\ldots, \tau_n)$ is in order of $R_g(s_0;\bar{m},\ldots,\bar{m})+O(\epsilon^2)$ for any $n\geq 2.$
\end{thm}

Due to the indistinguishability property one can  use result \ref{thmmain} to $r_g$ at each time $t$. Any time  $t\in \mathcal{T},$ $\|  r_g(s_t,a_t)-r_g(s_t,\bar{m}_t) \| $ is 
 bounded by $\delta_{t,\bar{m}_t,\bar{r}}(s_t)\| a_t-\bar{m}_t \|^2.$

Based on this we derive a bound on  the payoff $R_g.$  
\begin{eqnarray}
& R_g(s_0;\tau_1,\ldots, \tau_n)-R_g(s_0;\bar{m},\ldots,\bar{m})\\ \nonumber
&= g(s_T)-g(\bar{s}_T)+ \\ \nonumber &\sum_{t\in \mathcal{T}}r_g(s_{t},a_{t},m_t)-r_g(\bar{s}_{t},\bar{a}_{t},\ldots, \bar{a}_{t},m_t)\\ \nonumber
&= g(s_T)-g(\bar{s}_T)+  \\ \nonumber & \sum_{t\in \mathcal{T}}r_g(s_{t},a_{t},m_t)-r_g(\bar{s}_{t},{m}_{t},\ldots,{m}_{t},m_t)\\  \nonumber
&= g(s_T)-g(\bar{s}_T)+\sum_{t\in \mathcal{T}}r_g(s_{t},a_{t},m_t)-\bar{r}(\bar{s}_{t},{m}_{t})\\  
&= g(s_T)-g(\bar{s}_T)+\sum_{t\in \mathcal{T}}[\epsilon_t^2\bar{\delta}_{t,\bar{m}_t,\bar{r}}+o(\epsilon_t^2)].
\end{eqnarray}
where $\bar{s}_T$ is the final state when the average action $\bar{a}_{t}$ is played by all the minor players and 
$$\delta_{\bar{m},\bar{r}}(s_t)=|\frac{n}{2(n-1)}\left( -\frac{1}{n^2}\bar{r}''(\bar{m})+\partial^2_{a_1a_1}r_g(\bar{m}^{\bigotimes n})\right)|$$
and $\epsilon_t= \| a_t -m_t \|.$ 

Taking the absolute value, one obtains the following inequality:
 
 \begin{eqnarray}
 &\| R_g(s_0;\tau_1,\ldots, \tau_n)-R_g(s_0;\bar{m},\ldots,\bar{m}) \| \\
 &= \| \bar{g}(s_T)-\bar{g}(\bar{s}_T) \|+ \sum_{t\in \mathcal{T}}[\epsilon_t^2|\bar{\delta}_{t,\bar{m}_t,\bar{r}}|+o(\epsilon_t^2)]\\
 & \leq \delta_T+\sum_{t\in \mathcal{T}}[\epsilon_t^2\bar{\delta}_{t,\bar{m}_t,\bar{r}}+o(\epsilon_t^2)]  
 \end{eqnarray} where $\delta_T=\| \bar{g}(s_T)-\bar{g}(\bar{s}_T) \|.$  
 
 Now, a small changes in the action may change the state, and hence the term $\delta_{t,\bar{m}_t,\bar{r}}(s_t)$ is changed. Using the continuity of the state transition $q,$
  we take a uniform bound by considering the supremum: $\bar{\delta}_{t,\bar{m}_t,\bar{r}}=\sup_{s}\delta_{t,\bar{m}_t,\bar{r}}(s).$
  \begin{eqnarray}
 &\| R_g(s_0;\tau_1,\ldots, \tau_n)-R_g(s_0;\bar{m},\ldots,\bar{m}) \| \\
 & \leq \delta_T+[\sup_t \bar{\delta}_{t,\bar{m}_t,\bar{r}}]\sum_{t\in \mathcal{T}}\epsilon_t^2+o(\epsilon_t^2) 
 \end{eqnarray}
 For the same state, the error bound is $[\sup_t \bar{\delta}_{t,\bar{m}_t,\bar{r}}]\sum_{t\in \mathcal{T}}\epsilon_t^2.$ 
 Thus, the global error is bounded by $\epsilon^2 \left(\sup_{t\in \mathcal{T}}\bar{\delta}_{t,\bar{m}_t,\bar{r}}\right).$ 
 
 Since the  above inequality holds for a generic symmetric vector which is the average action, it is in particular true when evaluated at  a symmetric Nash equilibrium actions (if it exists) of the minor players.\footnote{The existence of symmetric equilibria in symmetric  stochastic games uses standard point fixed existence condition.} In that case, we recursively use the value iteration relation for each minor player
 $V_{j,t}(s_{j})=\sup_{a_j}\left[ r(s_j,a_j,m_t)+\mathbb{E}_{s'_j} V_{j,t+1}(s'_j\ | s_j,a_j,\tau) \right].$ 

Summing up over $j$ at the optimal strategies of the minor players (if any) yields 
$$V_{t}(s_{t})= r_g(s_t,a^*_t,m^*_t)+\mathbb{E}_{s'} V_{t+1}(s'\ | s_t,a^*_t,\tau^*).$$ Applying the local error bound $ [\epsilon_t^2\bar{\delta}_{t,\bar{m}_t,\bar{r}}(s_t)+o(\epsilon_t^2)]$ and iterating $T$ times 
gives
$$ R_g(s_0;\tau_1^*,\ldots, \tau^*_n)-R_g(s_0;\bar{m}^*,\ldots,\bar{m}^*) =O(\epsilon^2). $$

\section{Applications} \label{basicapplication}

\subsection{ Collaborative effort}
Consider $n$ players. Each player can choose an action in the closed interval $[0,1].$
The geometric aggregate given by
$\left( \prod_{j'=1}^n a_{j'}\right)^{\frac{1}{n}}=e^{\frac{1}{n}\sum_{j'=1}^n ln(a_{j'}) }$ over $[0,1]^n.$ In order to preserve the differentiability at the origin we consider the payoff as $r_j(a)=\prod_{j'=1}^n a_{j'}.$
Then, the  following statements hold:
{
\begin{itemize}
\item We observe that the payoff functions are indistinguishable. 
The payoff functions satisfy $r_j(a)=r_i(a)=r(a)=\left( \prod_{j'=1}^n a_{j'}\right),$ which remains the same by interchanging the indexes.

\item The pure strategies $0-$effort and $1-$effort are equilibria. Moreover, $1$ is a strong-equilibrium (resilience to any deviation of any size $k>1$). Indeed,

If all the players do the maximum effort, i.e., $a_j=1,$ then every player receives the maximum payoff $r_j=1$ and no player has incentive to deviate. This is clearly a pure Nash equilibrium. Suppose now that a subset of players (a coalition) deviates and jointly chooses an action that is different than $(1,\ldots,1)$, then the payoff of all the players is lower than $1.$ In particular, the members of the coalition gets a lower payoff than $1$. Since this analysis holds for any coalition of any size, the action profile $(1,\ldots,1)$ is a Strong Nash equilibrium.

\item We define the analogue of the price of anarchy (PoA) for payoff-maximization problem as the  ratio between the worse equilibrium payoff and the social optimum.  
If one of the players does $0-$effort (no effort) then the payoff of  every player will be zero and no player can improve its payoff by unilateral deviation.  This means that $(0,\ldots,0)$ is a pure Nash equilibrium.
Note that the equilibrium payoff at this equilibrium is the lowest possible payoff that a player can receive, i.e., $(0,\ldots,0)$ is the worse equilibrium in terms of payoffs.
 Hence, $PoA=0$ and the ratio between the  global optimum and the equilibrium payoff is $\frac{1}{0},$ which is infinite. 

 Clearly, the price of stability (the ratio between the best equilibrium payoff and the social optimum) is $ PoS=1 $.

\item We say that a pure  symmetric strategy $a^*$ is an evolutionarily stable strategy~\cite{price73} if it is resilient by small perturbation as follows:
For every $a\neq a^*,$ there exists an $\epsilon_{a}>0$ such that 

$ r((1-\epsilon)a^*+\epsilon a,\ldots,(1-\epsilon)a^*+\epsilon a, a^*, (1-\epsilon)a^*+\epsilon a, \ldots, (1-\epsilon)a^*+\epsilon a) >  $

$r((1-\epsilon)a^*+\epsilon a,\ldots,(1-\epsilon)a^*+\epsilon a, a, (1-\epsilon)a^*+\epsilon a, \ldots, (1-\epsilon)a^*+\epsilon a) $

for all  $\epsilon\in(0,\epsilon_{a}).$

We now show that  the pure strategy  $0-$effort (i.e., the action profile $(0,\ldots,0)$) is not  an Evolutionarily Stable Strategy.  Indeed, if  $a^*=0,$ the left hand side of the above inequality is $0$ which is not strictly greater than $a(\epsilon a)^{n-1}>0.$
\end{itemize}

The  non-asymptotic mean-field approach allows us to link the geometric mean with the arithmetic mean action.
We remark that the geometric mean  as a payoff, satisfies the indistinguishability property and it is smooth in the positive orthant. Here $\bar{r}$ is the identity function because when the all the actions are identical, the geometric mean coincides with the arithmetic mean. In order to illustrate the error bound in the non-asymptotic mean field let consider two decision-makers. We expand the payoff for an asymmetric input level of size $\epsilon.$ Let $\bar{m}=\frac{a_1+a_2}{2},\ \gamma_i=\frac{a_i-\bar{m}}{\epsilon}.$

\begin{eqnarray} r(a_1,a_2)&=& a_1a_2\\
&=& \bar{m}^2+\epsilon \underbrace{(\gamma_1+\gamma_2)}_{=0}+\epsilon^2\underbrace{\gamma_1 \gamma_2}_{abs. \leq 1}.
\end{eqnarray} 
Here, $\bar{r}(x)=x^2.$
Thus, $r(a)-\bar{r}(\bar{m})=r(a)-\bar{m}^2 =\epsilon^2\gamma_1 \gamma_2\leq \epsilon^2.$

In particular, If $a_1=1$ and $a_2=(1-\epsilon)^2<1.$
then $r=(1-\epsilon)^2$ and $\bar{m}=1-\epsilon+\frac{\epsilon^2}{2}.$ One has,
$$
\bar{m}^2-r(a)=\frac{\epsilon^4}{4}+(1-\epsilon)\epsilon^2< \epsilon^2.
$$

Now, if $a$ is near zero, i.e., $a_1=\eta$ and $a_2=\epsilon+\eta$ with $\eta\leq \epsilon/2.$ then
$r(a)=\eta^2+\epsilon\eta$ and $\bar{m}=\eta+\frac{\epsilon}{2}.$  Thus, 
$$  \bar{m}^2 - r(a)=\frac{\epsilon^2}{4}+\eta\epsilon \leq \frac{3}{4}\epsilon^2.$$
}
 The above calculus illustrates that if the system is indistinguishable we can work directly with the mean of the mean-field
with error $O(\epsilon^2)$ where $\epsilon$ captures the asymmetry level of the system.

Next we illustrate the usefulness of our approximation of waiting time in a queueing system with multiple servers.
\subsection{Queueing mean-field games} \label{queue example}
 Consider $n$ servers and a $M/M/n$ system with arrival $\lambda$ and service rate of $\mu_i$ for server $i.$
Assume that the  customers are indistinguishable in terms of performance index. 
 Each customer will be assigned to one of the non-busy servers with  a certain probability,
(if any). If not, the customer joins a
queue and will be in waiting list. Our goal is to investigate the delay, i.e., the propagation delay and  the expected
waiting time (WT) in the queue. Let the expected propagation delay to be $d_p.$  
Using~\cite{reftv0}, we determine the waiting time
$WT(\lambda,\mu_i=\bar m)
$ in the case  of similar service rates $\mu_i=\bar m$.
Let $\rho=\frac{\lambda}{n\bar m}.$ 
{
The transition rate (continuous time) is given by
$$\left\{ \begin{array}{c}
k\geq 1,\ R_{kk}= -[(k-1) \bar m+\lambda],\\
 R_{k,k+1}=\lambda, \\
k\geq 1, \ R_{k,k-1}=(k-1) \bar{m},\\
 \mbox{otherwise} \ R_{kj}=0.
 \end{array} \right.
 $$
 }
The steady states are easily determined by setting $Rx=0.$
 
%


The probability that all servers are busy is
$$C=\frac{\left( \frac{(n\rho)^n}{n!}\right) \left( \frac{1}{1-\rho} \right)}{\sum_{k=0}^{n-1} \frac{(n\rho)^k}{k!} + \left( \frac{(n\rho)^n}{n!} \right) \left( \frac{1}{1-\rho} \right)} .$$

Hence, the waiting response time for the symmetric setup $WT(\lambda,\mu_i=\bar m)=\frac{ C}{n\bar m - \lambda}.$

{
The computation of $WT(\lambda,\mu_1,\ldots,\mu_n)$ in the asymmetric setting is highly complex and is still not well understood. 
The question is to know if non-asymptotic mean field approach can provide a useful approximation of it.
To do so, we check the main assumptions $A0$ and $A1.$ 
 Clearly  WT is regular (in $\mu_1,\ldots, \mu_n$ for small $\lambda << \sum_{j=1}^n \mu_j$) and satisfies the indistinguishability property. Then, using nonasymptotic
 mean-field approach,
 \begin{align}
 WT(\lambda,(\mu_i)_i)
 &=WT(\lambda,\bar{m},\ldots, \bar{m})+O(\epsilon^2)\\
 &=\frac{ C}{n\bar m - \lambda}+O(\epsilon^2)
 \end{align}

\begin{figure}[htb]
  \centering
  \includegraphics[width=9cm,height=6cm]{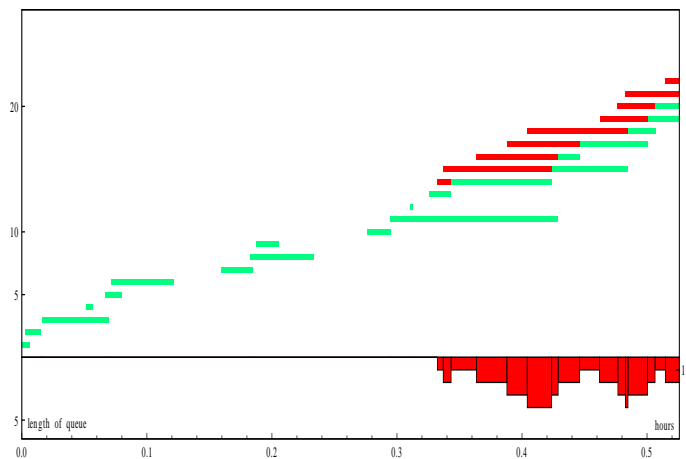}\\
  \caption{Simulation with two servers. The x-axis is the time evolution and the y-axis the queue size.}\label{figmm1simulation}
\end{figure}
In Figure \ref{figmm1simulation}, we observe  the following:

$
\begin{array}{|r|r|r|}
\hline
 \text{} & \text{expected} & \text{observed} \\
\hline
 \text{server utilization factor:  } \rho & 0.830 & 0.650 \\
 \text{ (proportion of time  } &  &  \\
\text{ each server is busy) } &  &  \\
\hline
 \text{probability of no customers }  & 0.091 & 0.150 \\
 \text{in the system}  &  & \\
\hline
 \text{probability that an arriving  }  & 0.760 & 0.460 \\
 \text{ customer has to wait } &  &  \\
\hline
 \text{average number of customers }  & 3.800 & 0.810 \\
 \text{in the queue}  &  &  \\
\hline
 \text{average number of customers } & 5.500 & 2.100 \\
  \text{in the system}  &  &  \\
\text{including the queue} &  &  \\
\hline
\end{array}
$
}
\subsection{ Auction with  asymmetric bidders}
\subsubsection{Static setup}
The theory of auctions as games of incomplete information originated in
1961 in the work of Vickrey.
A seller has an object to sell. She adopted a first-price auction rule.
Consider a first-price auction with    asymmetric bidders. There are $n\geq 2$ bidders for the object.
 Each bidder independently submit a single bid without seeing the others' bids.
 If there is only one bidder with the highest bid, the object is sold to the bidder with biggest bid. The winner pays her bid, that is, the price is  the highest  (or first price bid).
 If there is more than one bidder, the object goes to each of these bidders with equal probability.
 The bidder $v_j$ has a valuation of the object. The random variable $\tilde{v}_j$ has a $C^1-$cumulative distribution function with support $[\underline{v},\bar{v}]$ where  $\underline{v}< \bar{v}.$
 A strategy of bidder $j$ is a mapping from valuation to a bid space: $v_j   \longmapsto b_j(v_j).$
 The risk-neutral payoff of bidder $j$ is   $(v_j-b)\mathbb{P}\left( \max_{j'\neq j} {b}_{j'}(\tilde{v}_{j'}) < b \right).$
Using the independence of the valuation $\tilde{v}_{j'}$, the risk-neutral payoff can written as
$(v_j-b)\prod_{j'\neq j} F_{j'}(b_{j'}^{-1}(b)).$
The information structure of auction game is as follows. Each bidder knows its value, bid but not the valuation of the other bidders. Each bidder knows the valuation cumulative distribution of the others. The structure of the game is common knowledge.
We are interested in the equilibria, equilibrium payoffs and revenue of the seller.
Existence of equilibrium of auction games have been widely studied (\cite{reftv1,lebrun1,lebrun2}).

Clearly, no bidder would bid an amount that is greater than her value because of negative payoff.
By fixing the bidding strategy of the others one has attempted to compute the best response correspondence. Any increase in the bid will decrease the gain but increase the probability of winning. This is a sort of tradeoff between the profit and the probability of winning.

We differentiate the function $b     \longrightarrow (v_j-b)\prod_{j'\neq j} F_{j'}(b_{j'}^{-1}(b)).$

 In order to find an equilibrium one needs  to
 solve $n$ Ordinary Differential Equations (ODEs) with boundary conditions.
\begin{eqnarray} \label{conditiont2}
v'_j(b)=\frac{F_j(v_j(b))}{F'_j(v_j(b))}\left[  \frac{1}{n-1}\sum_{j'=1}^n\frac{1}{v_{j'}(b)-b}-\frac{1}{v_{j}(b)-b}\right]
\end{eqnarray}
\begin{eqnarray}\label{conditiont3}  v_j(\underline{b})=\underline{v},\ v_j(\bar{b})=\bar{v}
\end{eqnarray}  The inverse of the function $v$ is the optimal strategy $b.$
There is no need to mention that this is intractable even with small number of bidders. Even for three bidders we do not understand clearly how the solutions behave in function of $F_j.$

Why this is not a simple ODE problem?

{\bf Non-standard existence theorem is needed:}  We cannot apply the standard local existence and uniqueness theorem to the ODE with initial value (lowest bid)
$ v_j(\underline{b})=\underline{v}$
because by  the right-hand-side terms $\frac{1}{v_{j}(b)-b}$ in the ODEs are unbounded at $\underline{v}.$
In addition, the equilibrium satisfies  $v_j(\bar{b})=\bar{v}$ but the term $\bar{b}$ is unknown.
Due to these difficulties, explicit solutions of (\ref{conditiont2})  and  (\ref{conditiont3}) are not available.

{\bf Non-standard numerical method is needed:}

Since explicit solutions are open issues, one may ask if it is possible to solve the problem numerically. According to the recent work in \cite{dynamict1}, the numerical implementation of the system (\ref{conditiont2}) ,  (\ref{conditiont3})  remains a challenging task.
One of the well-known numerical methods consists to solve to find among the solutions of ODEs together with the initial conditions, that satisfy the highest equilibrium bid constraint. Such an approach is known as {\it forward-shooting method}. However,
 the forward-shooting method of Marshall et al. \cite{marshall} do not converge to the solution due to approximation near $\underline{b},$ with the derivative
$
v'_j(\underline{b})
$
 It has been shown in \cite{dynamict1} that for the special case of power law (i.e. $F(v)=v^{\alpha}$), a dynamical system approach can be used with the change of variable
$$ v_j(b)=b V_j(b),\ b=e^{w},
$$

In the backward approach, one searches for the value of $\bar{b}$ by solving Equation (\ref{conditiont2}) backward in $b$
 subject to the end condition $v_j(\bar{b})=\bar{v}$
and looking for the value of $\bar{b}$ for which the initial value coincides with $\underline{v}.$
However, the standard backward-shooting method is inherently unstable, specially when the bids are near $\underline{v}.$ The authors in \cite{ganishnum} showed that  the  backward-shooting method is unstable even in the symmetric case.

If all the functions $F_j(v)$ are the same, and hence equal to $\bar{m}(v)$ then,
 we know from Vickrey 1961 that  the symmetric equilibrium is $$b(s)=s-\frac{\int_{\underline{v}}^{s} F^{n-1}(x) dx}{F^{n-1}(s)},$$
 which is obtained as follows:

 Instead of $n$ ODEs we have one ODE to solve. The ODE is
\begin{eqnarray}v'(b)(v(b)-b)=\frac{G(v(b))}{G'(v(b))}
\end{eqnarray}
where $G$ the value distribution of the $n-1$ bidders.
Using the bijection function and the fact that $(h^{-1})'(x)=\frac{1}{h'(h^{-1}(x))}$
Hence,
$\frac{1}{h'(x)}(x-h(x))G'(x)=G(x)$ where $h(x)$ is the strategy.
This means that
$xG'=h'G+hG'=(hG)'.$
By simple integration between the minimum value and $v,$ one gets
$$
h(v)G(v)-[hG]_{\underline{v}}=\int_{\underline{v}}^v x G'(x) dx=[xG(x)]^{v}_{\underline{v}}-\int^v_{\underline{v}} G(x)\ dx
$$
Hence, $h(v)=v-\frac{1}{G(v)}\int^v_{\underline{v}} G(x)\ dx$

{\it Nonasymptotic mean field approach provides a useful error bound in this open problem}

For asymmetric distribution  we are able to get a precise error bound when the distribution $F_j$ are close to their arithmetic mean, the equilibrium strategies and payoffs can be approximated in a perturbed range. To do so, we use a non-asymptotic mean field approach over function space.
First remark that the revenue of the seller, satisfies the indistinguishability property, since it is, up to a constant, the integral of the product $\prod_{j\in\mathcal{N}}F_j.$
We rewrite the function $F_j$ as  $F_j(v)=\bar{m}(v)+\epsilon \gamma_j(v),$ where
$ \gamma_j(v)=\frac{F_j(v)-\bar{m}(v )}{\epsilon},$
$\bar{m}(v)=\frac{1}{n}\sum_{j'=1}^n F_{j'}(v)$ and
$\epsilon=\max_j \max_{[\underline{v}, \bar{v}]}\ | F_j(v)-\bar{m}(v ) |,$

Using result \ref{thmmain}, one gets
\begin{itemize}
\item Good approximate of the asymmetric equilibrium strategies,
\item Equilibrium payoff with deviation order of $O(\epsilon^2).$
\end{itemize}

\begin{example}
 Note that the  Optional Second Price  auction is currently used in Doubleclick Ad Exchange. Examples of Ad exchanges are
RightMedia, adBrite, OpenX, and DoubleClick. The idea is described as follows \cite{adx,nisan}:
\begin{itemize}
\item  User $u$ visits the webpage $w$ of publisher $p(w)$ that has, say, a single slot for ads.
\item  Publisher $p(w)$ contacts the exchange E with $(w, P(u), \rho)$ where $\rho$ is the minimum
price p(w) is willing to take for the slot in $w,$ and $P(u)$ is the information about user
$u$ that $P(w)$ shares with $E.$
\item The exchange E contacts ad networks $adn_1,\ldots, adn_m$ with $(E(w), E(u))$, where $E(w)$ is
information about $w$ provided by $E$, and $E(u)$ is the information about $u$ provided by
$E.$  $ E(u)$ may be potentially different from $ P(u)$.
\item  Each ad network $adn_j$ returns $(b_j, d_j) $ on behalf of its customers which are the advertisers; $b_j$
is its bid, that is, the maximum it is willing to pay for the slot in page $w$ and $d_j$
is
the ad it wishes to be shown. Each ad network may have  multiple advertisers. The ad networks may also choose not to return a bid.
\item  Exchange $E$ determines a winner $j^*$
for the ad slot among all $ (b_j, d_j)'$s and its price $c_{j^*}$ satisfying $\rho\leq c_{j^*}\leq b_{j^*}$
via an auction (first or second price).
\item Exchange $E$ returns winning ad $d_{j^*}$
to publisher p(w) and price $c_{j^*}$ to ad network $j^*$
\item The publisher $p(w)$ serves webpage $w$ with ad $d_{j^*}$
 to user $u$ (the
impression of ad $d_{j^*}$).
\end{itemize}
Note that from the click of the user to the impression of the ad $d_{j^*},$ there are many intermediary interactive processes. Auction is one them and an important one because it determines the winner ad network. While it is reasonable to consider large population of users over internet, the number of concurrent ad networks remains finite and there a room for non-asymptotic mean-field analysis for the revenue. As a user may click several times over  webpages, the dynamic auction framework seems more realistic. We examine the dynamic auction  in subsection~\ref{subsecdynamic}.
\end{example}

\subsubsection{Spiteful bidders}
A player might be losing the auction of a long-term project. Yet she
continues to participate in the auction because she wants to minimize the negative payoff on losing by making
her competitor, who would win the auction, pay a high price for the win. This negative dependence of payoff
on others' surplus is referred to as spiteful behavior. Below we show how our nonasymptotic mean-field framework  can be applied to that scenario.

A spiteful player $j$ maximizes the weighted difference of her
own payoff $r_j$ and his competitors' payoffs $r_{j'}$ for all $j'\neq j.$
The payoff of a spiteful player is
\begin{eqnarray}
r_{j,\alpha}:=(1-\alpha_j)r_j-\alpha_j\sum_{j'\neq j}r_{j'}
\end{eqnarray}
Obviously, setting $\alpha_j$
to zero yields a selfishness
 (whose payoff equals his exact profit) whereas $\alpha_j = 1$ defines
a completely malicious player (jammer) whose only goal is to minimize
the profit of other players. Note that for altruistic player we would be considering a payoff in the form
$
r_{j,\alpha}=(1-\alpha_j)r_j+\alpha_j\sum_{j'\neq j}r_{j'}.
$

The payoff of a spiteful player $j$ is
\begin{eqnarray}  \nonumber
r_{j,\alpha}&=&(1-\alpha_j)(v_j-b_j)\prod_{j'\neq j}F_{j'}(v_{j'}(b_j))\\ & & \nonumber
-\alpha_j \mathbb{E}\left[\max_{j'\neq j}v_{j'}| \ \max_{j'\neq j}b_{j'} > b_{j}(v_j) \right]\\
&&
+\alpha_j \mathbb{E}\left[\max_{j'\neq j}b_{j'}| \ \max_{j'\neq j}b_{j'} > b_{j}(v_j) \right]
\end{eqnarray}

As mentioned above the main difficulty is that the private values distribution are asymmetric.
Denote that by $b^*_{j,\alpha}(v)$ the equilibrium bid strategy.
 Even when bidders are selfish ($\forall j,\ \alpha_j = 0$), the above analysis shows that the explicit expression of $b^*_{j,0}(v)$ is NOT a trivial task.  However, for symmetric type distribution, and symmetric coefficient $\alpha_j=\alpha,$ the payoff function reduces to
 $$
 r_{j,\alpha}=(1-\alpha)(v_j-b)F^{n-1}(v)-\alpha \int_{v_j}^{\bar{v}} s (n-1)F^{n-2}(s)f(s) ds$$ $$
 -\alpha \int_{b_j(v_j)}^{b_j(\bar{v})} (n-1) s F^{n-1}(v(s))f(v(s))v'(s) ds
 $$

 Using the fact that the derivative of $\int_{v(b)}^{\bar{v}}h(s)\ ds$ with the respect to $b$ is given by $-h(v(b))v'(b).$ The first order optimality condition yields to
 $$
 b(v)=v-\frac{(1-\alpha)F(v)}{(n-1)f(v)}b'(v).
 $$

 In particular for uniform distribution $F_j(v)=v$ over $[0,1]$   the  Bayes
Nash equilibrium strategy  for spiteful player satisfies the ODE:
 $$
 b_{\alpha}(v)=v[1-\frac{(1-\alpha)}{(n-1)}b'(v)].
 $$

It is easy to check that $b(v)=\frac{n-1}{n-\alpha} v$ is a solution of the above ODE.
 The symmetric equilibrium is increasing and convex with $\alpha.$ We now evaluate the revenue of the seller in equilibrium for $F_j(v)=\bar{m}(v)=v$.
\begin{eqnarray}
R_{\alpha}(\bar{m},\bar{m}))&=&
\mathbb{E}[ \max_{j}b_j(v_j)]\\
&=& \int b(v) n F^{n-1}(v) f'(v)dv\\ \nonumber 
&=& \int n (\frac{(n-1)}{(n-\alpha)}v)  v^{n-1}\ dv\\ \nonumber
&=& \frac{n(n-1)}{(n-\alpha)(n+1)}
\end{eqnarray}

Now we compare the revenue with $R_{\alpha}(F_1,F_2)$ where
$$
F_1(v)=\bar{m}(v)-\epsilon v(1-v)(\frac{1}{2}-v)
$$
and $$
F_2(v)=\bar{m}(v)+\epsilon v(1-v)(\frac{1}{2}-v)
$$

$$
R_{\alpha}(F_1,F_2)=R_{\alpha}(\bar{m},\bar{m})+O(\epsilon^2).
$$

{

\subsection{Fast algorithm for computing approximate equilibrium}
We construct a fast algorithm for computing approximate equilibrium.
Recall that the first optimality equation can be written as

$$
1+(b-v_{j}(b))\sum_{j'\neq j} \frac{F'_{j'}(v_{j'}(b))}{F_{j'}(v_{j'}(b))} v'_{j'}(b)=0
$$

Define the functional 
$H_j(\bar{b},v)=1+(b-v_{j}(b))\sum_{j'\neq j} \frac{F'_{j'}(v_{j'}(b))}{F_{j'}(v_{j'}(b))} v'_{j'}(b).$

We consider polynomial expansion of inverse-bid functions.   The function $v_j$ is written in a flexible
functional form $v_j(b) =\bar{b}-\sum_{k=0}^{+\infty} \mu_{j,k} (\bar{b}-b)^{k}.$ 

We truncate this polynomial to order $K\geq 2$ and replace it in the first order optimality equation.
Denote $\hat{v}_{j,K}(b)=\bar{b}-\sum_{k=0}^{K} \mu_{j,k} (\bar{b}-b)^{k}.$ 
Taking in account $2n$ boundary conditions, one gets that

$$
L(\bar{b},v)=\sum_{j=1}^n H_j(\bar{b},v)^2+\sum_{j=1}^n (v_j(\bar{b})-\bar{v})^2 +\sum_{j=1}^n (v_j(\underline{b})-\underline{v})^2\geq 0
$$
has a minimum $0$ and the minimizer is the equilibrium  inverse bid strategy $v.$
Hence, it is reasonable  to consider the functional $H$ when each of function $\hat{v}_{j,K}$ belongs the subspace $\mathcal{D}_K$ the set of polynomial with degree at most $K.$ This is space with dimension $K+1.$ The problem becomes
$$\inf_{(\hat{v}_{j,K})_j \in \mathcal{D}_K }L(\bar{b},\hat{v})= \inf_{(\mu_{j,k})_{j,k}}L(\bar{b},\hat{v})$$

Remember that $\bar{b}$ is the highest bid that
is submitted in equilibrium. It is 
therefore an unknown. Thus, we add this into the optimization problem. Hence one has $1+n(K+1)$ unknown variables to find.
Using a grid decomposition of the domain $[ \underline{v}, \bar{v}]$ with $T$ points inside, we arrive at a
nonlinear least-squares algorithm for  selecting $\bar{b}$ and $(\mu_{j,k})_{j,k}$
by solving 
$\inf_{(\hat{\bar{b}}_t)_t, (\mu_{j,k})_{j,k}}\sum_{t=1}^T L(\hat{\bar{b}}_t,\hat{v}_K),$
which yields

$$\inf_{(\hat{\bar{b}}_t)_t, (\mu_{j,k})_{j,k}}\sum_{t=1}^T \sum_{j=1}^n H_j(\hat{\bar{b}}_t,\hat{v}_K)^2 +$$ $$
T\sum_{j=1}^n (\hat{v}_{j,K}(\bar{b}_t)-\bar{v})^2 +T\sum_{j=1}^n (\hat{v}_{j,K}(\underline{b})-\underline{v})^2$$

The points $\hat{\bar{b}}_t$ on a grid  will be chosen  uniformly
spaced, i.e. $\hat{\bar{b}}_t= \bar{b}+\frac{t}{T} (\underline{v}-\bar{b})$

Standard Newton-Gauss-Seidel methods provide a very fast convergence rate to a solution if the initial guess if appropriately chosen. However the choice of initial data and guess need to be conducted.

We propose a numerical scheme for optimal bidding strategies.
First we solve the initial-value problem that starts at $b_0$ near but not equal to $\underline{b}$ so that the denominator do not vanish. ODE starts at $b_0>\underline{b}$ and moves forward. We fix the starting function to
$$v(b_0)=\underline{v}+(b-b_0)D_1+ \delta_1 D_2 (b-b_0)^{1+\delta_2}$$

where $D_1=(D_{11} ,D_{12})$ is the limit in $\underline{b}$ of the derivative, $D_{11}=\lim_{b\longrightarrow b_{0}} \frac{v_1(b)-v_1(\underline{b})}{b-\underline{b}},$
$D_{12}=\lim_{b\longrightarrow b_{0}} \frac{v_2(b)-v_2(\underline{b})}{b-\underline{b}},$
, $\delta_i$ are positive constant.
and $D_2=[D_{21},D_{22}],$ is the equivalent fractional derivative approximation around $\underline{b},$ $b_0=\underline{b}+10^{-\frac{7}{1+\delta_2}}$
 To find the solution we take the intersection $v_1(b_+)=v_2(b_+)=x$
 The solution is $(\frac{1}{x}v_1(b x),\frac{1}{x}v_2(b x)).$

\begin{figure}
  \centering
  \includegraphics[width=9cm]{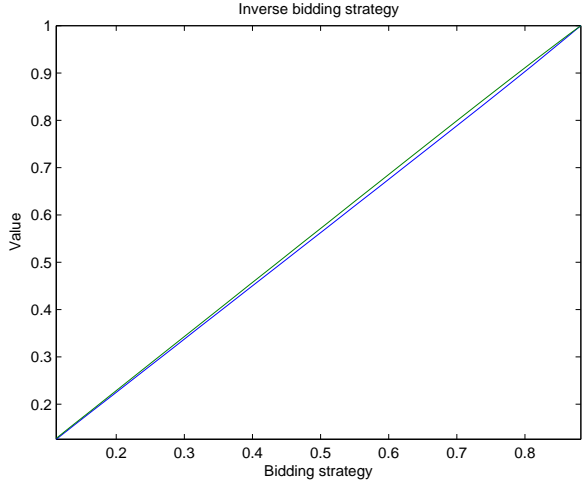}\\
  \caption{inverse optimal strategy for the distributions $F_1(v)=v^{7},$
$F_2(v)=v^{8}.$ The x-axis represents the optimal bidding strategy and the y-axis is the inverse optimal bidding strategy of the two-players.}\label{fignoauction1}
\end{figure}
Figure \ref{fignoauction1}  represents the inverse optimal bidding strategy for the cumulative distributions $F_1(v)=v^{7},$ $F_2(v)=v^{8}.$ The x-axis represents the optimal bidding strategy and the y-axis is the inverse optimal bidding strategy of the two players. For that case $\epsilon=\frac{1}{8}\left(\frac{7}{8}\right)^7.$ We observe in Figure \ref{fignoauction1} that the second curve is around the first one plus  $\frac{1}{64}\left(\frac{7}{8}\right)^{14}.$

  For $n\geq 2$ we do not understand yet the behavior of $b^*_{j,\alpha}(v)$. However, we are able to provide a useful approximation for asymmetric distribution in  function of their deviation to the mean. Moreover the approximation holds  for the revenue of the seller (auctioneer).
}
\subsubsection{Dynamic auction with asymmetric bidders} \label{subsecdynamic}
We now explain how our framework can be extended in a dynamic setting.
What is the motivation for dynamic auction?
Currently a large proportion of internet users employ search engines to
locate information~\cite{alex}. For example, search engines are used to obtain an instantaneous selection of offers from various providers operating in the cyberspace.
Internet auctions, on the other hand, connect potential sellers and buyers from
different locations in a virtual auction.
Buyers also often get multiple purchase opportunities over time. This is most obvious in the
case of today's online auction markets: eBay has auctions for physical goods closing every
second, while Google's Doubleclick and Microsoft's Advertising Exchange trade online advertisements at a much faster rate. This creates the possibility for buyers to inter-temporally
substitute, adjusting their current bids to account for the option value of waiting for future
purchasing opportunities.

In sponsored search auctions, such as those run by Google (AdWords) or Bing (adCenter), a query to a search engine triggers an auction for positions on the page returned by the query. Advertisers bid for positions or slots, and successful bids result
in ads being displayed alongside results to the query. The slots are ordered,
where higher slots are more valuable. Under a dynamic auction mechanism for slots, an agent places a single bid, bids are ranked by weight and
ranked bids determine the slot allocation. If an advert displayed in a slot is clicked,
then the advertiser pays the price to be in the current slot, a Pay-Per-Click payment model.

In order to capture the practical observation, we model the problem as dynamic auction between buyers and sellers. We specially focus on the buyers side. Since measurement are done in discrete time unit, we start with  a discrete time model.
The work in \cite{alex} considered repeated second-price auctions under homogeneous valuation distributions and  budget constraints.  Most auctions involve bidders that are heterogeneous ex ante. For example  ad networks may not have homogeneous valuation distribution for the advertisers. Therefore, Here we consider asymmetric distribution as well.

 Time space is $\mathcal{T}=\{0,1,\ldots,T-1\}$ where $T$ is the length of the horizon. Since budget is limited for the entire day or month, we transform the long-term budget constraints into an iterative budget state equation for bidder $j:$

$s_{j,t+1}=s_{j,t}-c(s_{j,t},b_{j,t})\ind_{\{ b_{j,t}> \max_{j'\neq j}b_{j',t}\}}$ where $s_{j,t}$ is the remaining budget  of bidder $j$ for the corresponding contract frame, $c(s_{j,t},b_{j,t})$ is the total cost for the bid (if winner). The state $s_{j,t}$ is subject to positivity constraint: $s_{j,t}\geq 0$ almost surely. It clearly limits the action space at time $t$ to the set of bids $b_{j,t}$ such that $c(s_{j,t},b_{j,t})\leq s_{j,t}.$ Let $g$ be the function $b \longmapsto c(s_{j,t},b).$
The individual state bidder $j$ is the pair $(s_{j,t},v_{j,t}).$ This generates a stochastic game with Markovian transition. Each bidder maximizes her long-term payoff given by
$$R_{j,T}=\mathbb{E}\left[\bar{g}(s_T) \right. $$ $$\left.+\sum_{t\in\mathcal{T}}(v_{j,t}-b_{j,t})\ind_{\{ \max_{j'\neq j} {b}_{j',t}(v_{j,t}) < b_{j,t}(v_{j,t})\}} | \ v_j\right]$$

A history of a bidder $j$ is at time $t$ is a collection $(s_{j,t'},v_{j,t'},b_{j,t'})_{t'\leq t}.$
A bidding strategy $\tau_{j,t}$ of $j$ at time $t$ is a mapping from the history $h_{j,t}\in \mathcal{H}_{j,t}$ to the restricted bid space $[\underline{b},\bar{b}]\bigcap [0, g^{-1}(s_{j,t})].$

The game is played as follows.
At opportunity $t\in\mathcal{T},$ every player $ j$, $j\in\mathcal{N},$
\begin{itemize}\item realizes his current value $v_{j,t}\in [\underline{v}, \bar{v}]$ distributed according to $F_{j,t}$
\item submits a bid $b_{j,t} = \tau_{j,t}(v_{j,t}, h_{j,t}),$ where
$\tau_{j,t}:\  [\underline{v}, \bar{v}]\times \mathcal{H}_{j,t}    \longrightarrow  [\underline{b}, \bar{b}]$
denotes his bidding strategy at auction $t$;
\item updates his information set \ $ h_{j,t+1} \in \mathcal{H}_{j,t+1},$ based on the results obtained in auction $t.$ Specifically, he forms a
set of beliefs about the distribution of the bidding profile of the players  $b_{-j,t+1}$ distributed according to $F_{j',t+1},\ j'\neq j.$
\end{itemize}

Let $V_{j,0}(s_0,T,\tau_{-j})$ be the value function of the bidder $j$ , i.e., it is the supremum, over all
possible bidding strategies, of the expectation of the payoff $R_{j,T}$ starting from
an initial budget $s_0$ when the other bidder strategy profile is $\tau_{-j}.$

Based on the classical Bellman optimality criterion, we immediately get the following result. The proof  is therefore omitted.
\begin{thm} \label{refdynamic}
Given $v_j$ the optimal strategy of a player $j$ satisfies
\begin{eqnarray}\nonumber b_{j,t}^*(v_j)&\in &
\arg\max_{b\leq g^{-1}(s_{j,t}) }\mathbb{E}_{\bar{b}_{-j,t}}\left[  (v_j-b)\ind_{\{ b> \bar{b}_{-j,t}\}}
\right.\\ \nonumber & & \left. +
\ind_{\{ b> \bar{b}_{-j,t}\}}V_{j,t+1}(s_{j,t}-c(s_{j,t},b))
\right. \\  & &  \left.
+\ind_{\{ b< \bar{b}_{-j,t}\}}V_{j,t+1}
\ | \ v_j\right],
\end{eqnarray}
where $\bar{b}_{-j,t}:=\max_{j'\neq j} {b}_{j',t}$ is the highest bid of the other players than $j.$

Let $x^*_1\in\arg\max_{x_1\in [\bar{b}_{-j,t},g^{-1}(s_{j,t})]} $ $ \left[v_j-x_1+V_{j,t}(s_{j,t}-c(s_{j,t},x_1))\right] .$

Let $W_j:=v_j-x_1^*+V_{j,t}(s_{j,t}-c(s_{j,t},x^*_1)),$
Then the optimal bidding strategy is
$$b^*_{j,t}(v_j)=\left\{\begin{array}{cc}
x^*_1& \mbox{if} \ W_{j,t}> V_{j,t+1}(s_{j,t})\\
x^*_2\leq \min\left(  \bar{b}_{-j,t},g^{-1}(s_{j,t}) \right) &  \mbox{otherwise}
\end{array}
\right.
$$

The value iteration is given by
\begin{eqnarray}\nonumber
V_{j,t}(s_{j,t})&=&
\mathbb{E}_{v_j}\sup_b \mathbb{E}_{\bar{b}_{-j,t}}\left[  (v_j-b)\ind_{\{ b> \bar{b}_{-j,t}\}}\right.\\ \nonumber
& &
+
\ind_{\{ b> \bar{b}_{-j,t}\}}V_{j,t+1}(s_{j,t}-c(s_{j,t},b))\\  & &
\left.+\ind_{\{ b< \bar{b}_{-j,t}\}}V_{j,t+1}(s_{j,t})
\ | \ v_j\right]
\end{eqnarray}

\end{thm}

For the symmetric setup we drop the index $j.$
\begin{thm}[Symmetric beliefs]
Given $v$ the optimal strategy of a generic player  satisfies
\begin{eqnarray} \nonumber
b_{t}^*(v)&\in & \arg\max_{b\leq g^{-1}(s_{t}) }\mathbb{E}_{\bar{b}_t}\left[  (v-b)\ind_{\{ b> \bar{b}_t\}} \right. \\ \nonumber && \left. +
\ind_{\{ b> \bar{b}_t\}}V_{t+1}(s_{t}-c(s_{t},b))\right. \\  & &  \left.
+\ind_{\{ b< \bar{b}_t\}}V_{t+1}(s_{t})
\ | \ v\right],
\end{eqnarray}
where $\bar{b}_t:=\max_{j'\neq j} {b}_{j',t}$ is the highest bid of the other players than $j.$

Let $x^*_1\in\arg\max_{x_1\in [\bar{b}_t,g^{-1}(s_{t})]} \left[v-x_1+V_{t}(s_{t}-c(s_{t},x_1))\right] .$
This implies that the optimal bidding strategy is
$$b^*_{t}(v)=\left\{\begin{array}{cc}
x^*_1& \mbox{if} \ W_t> V_{t+1}(s_{t})\\
x^*_2 < \min\left(  \bar{b}_{t},g^{-1}(s_{t}) \right) &  \mbox{otherwise}
\end{array}
\right.
$$

The value iteration is given by
\begin{eqnarray} \nonumber
V_{t}(s_{t})&=&\mathbb{E}_{v}\sup_b \mathbb{E}_{\bar{b}_{t}}\left[  (v-b)\ind_{\{ b> \bar{b}_{t}\}}
\right.\\ \nonumber  & & +
\ind_{\{ b> \bar{b}_{t}\}}V_{t+1}(s_t-c(s_t,b))\\ && \left.+\ind_{\{ b< \bar{b}_{t}\}}V_{t+1}(s_{t})
\ | \ v\right]
\end{eqnarray}
\end{thm}
\begin{proof}
The proof follows as a corollary of result \ref{refdynamic}.
\end{proof}
To complete the value iteration system we choose a terminal payoff $V_{T}(s)=\bar{g}(s).$
\begin{thm}[Non-asymptotic mean-field] \label{reyu}
Let $F_{j,t}(v)=\bar{m}_t(v)+\epsilon \gamma_{j,t}(v),$ where
$$ \gamma_{j,t}(v)=\frac{F_{j,t}(v)-\bar{m}_t(v )}{\epsilon},$$
$$\bar{m}_t(v)=\frac{1}{n}\sum_{j'=1}^n F_{j',t}(v),$$ and
$$\epsilon=\max_j \sup_{t\in \mathcal{T}}\sup_{[\underline{v}, \bar{v}]}\ | F_{j,t}(v)-\bar{m}_t(v ) |,$$

Then, the long-term revenue of the seller $R_0(s_0;F_1,\ldots, F_n)$ is in order of $R_0(s;\bar{m},\ldots,\bar{m})+O(\epsilon^2)$ for any $n\geq 2.$
\end{thm}

\begin{proof}
The proof of the first statement is a direct extension of result \ref{thmmain} to the time space $\mathcal{T}.$ 
\end{proof}
\begin{rem}
As a consequence of the Result \ref{reyu},  if $F_{j,t}=F_{j,t}^n$ depends on $n$ and satisfies ${n}^{\alpha}(F_{j,t}^n-\bar{m}_t(v)) \leq \delta$ for some $\alpha>0$ then,
the error gap between the finite regime and the infinite regime in equilibrium is in order of $O(\frac{1}{n^{2\alpha}})$ which can be very small even for small $n$ but large $\alpha.$

If  ${n}^{\alpha}(F_{j,t}^n-\bar{m}_t(v)) \leq \delta$ then the error gap at time $t$ reduces to $\delta_{t,\bar{m}_t,\bar{r}}\frac{\delta}{{n}^{2\alpha}}$ and hence the global error in $\mathcal{T}$ is at most
$$
\frac{\delta}{{n}^{2\alpha}}\left(\sup_{t\in \mathcal{T}}\delta_{t,\bar{m}_t,\bar{r}}\right).
$$
Note that is a significant improvement of the  mean-field approximation since the use of mean-field convergence \`a la {\it de Finetti}~\cite{dyna2} gives a convergence order of $O(\frac{1}{\sqrt{n}}).$ The use of the indistinguishability property of the payoff function helps us to provide a more precise error.
\end{rem}

\subsection{Scaled auction problem}
Using a scaling factor $\frac{1}{\lambda},$ to the starting state (budget) and horizon, the value
$$W_{0}^{\lambda}(s,T):=\lambda V_{0}(\frac{s}{\lambda},\frac{T}{\lambda}), $$
has a certain limit when $\lambda$ goes to $0.$  Let
$\lim_{\lambda} W_{0}^{\lambda}(s,T)=w(s,T).$
Let $T^*=\inf\{ t\geq 0\  \ | \ s_t\leq 0, s_0>0\}.$
Then, the value $w$ is solution of the following differential game
$$w(s,T)=\sup_{\tau}\ \int_0^{\inf(T^*,T)}
\mathbb{E}_{v,\bar{b}} \left[(v_t-b_t)\ind_{\{b_t>\bar{b}_t\}}\ \right] dt $$
subject to
$$\dot{s}_t=-\mathbb{E} c(s_t,b_t),$$
$T^*=\inf\{ t\geq 0,\ | \ s_t=0 , s_0>0\},$
$w(0,T)=0$ and $w(s,0)=0.$
Let $r(b)$ be the instantaneous payoff from the above formulation.
Introduce the Hamiltonian
$$ \bar{H}(s,p)=\sup_{b}\left[ r(b)-p c(s,b) \right],$$

The value $w$ satisfies the Hamilton-Jacobi-Bellman equation
$$
\partial_T w+\bar{H}(s,\partial_sw)=0
$$
{
\subsection{Continuous time game with incomplete information}
We consider a particular state dynamics given by drift $f$, independent individual Brownian motion $\sigma_i$ and a common Brownian motion $\sigma_c$. The instantaneous cost is $c(t,x,u,m)$
Define the Hamiltonian as $H(t,x,p,m)=\inf_{u} \{ c + \langle p, f\rangle \}$ and let $v(t,x,m)$ be an equilibrium cost value. Following \cite{pierrelouis}, the value satisfies
$$
\left\{
\begin{array}{c}
\partial_t v-(\sigma_i+\sigma_c) \Delta_xv +H(t,x,\partial_xv, m)\\
+ \langle \partial_m v; -(\sigma_i+\sigma_c)\Delta_x m+div_x(m f(t,x,u^*,m)) \rangle\\
-\sigma_c \langle \partial^2_{mm}v\ \partial_{x}m; \partial_{x}m\rangle+ 2\sigma_c \langle \partial_m\partial_xv; \partial_x m \rangle =0\\
v(0, x,m)=g(x,m)\\
\end{array}
\right.
$$
where $div_x$ is the divergence operator.
Considering $n$ minor players where the evolution of the measure is now replaced by the evolution of beliefs in a Bayesian mean field game, the long-term cost function of a player
is  $R(\tau_1,\ldots, \tau_n)$ which in order of $R(\hat{m},\ldots,\hat{m})+O(\epsilon^2)=v(T,x,\hat{m})+O(\epsilon^2)$ where $\hat{m}$ is the average measure (belief). Thus, the non-asymptotic mean field game approach allows us to understand the behavior of the equilibrium cost when the players have different beliefs (incomplete information game) which are near the average belief measure.

\subsection{Extensions}
\subsubsection{Near-indistinguishable games}
In this section we assume that the payoff functions are strictly positive and satisfy a near-indistinguishability property defined as follows:

\begin{defi}[Non-scalable notion]
The game $G$ with payoffs $r_j$ is near-indistinguishable if it is $\epsilon-$indistinguishable for  a certain small $\epsilon>0$  i.e. there exists an indistinguishable function ${r}_g$ such that $\sup_a | r_j(a)-{r}_g(a)|\leq \epsilon $
\end{defi}

\begin{defi}[Scalable  notion]
The game $G$ with payoffs $r_j$ is scalable near-indistinguishable if it is $\epsilon-$scalable indistinguishable for  a certain small $\epsilon>0$  i.e. there exists an indistinguishable function ${r}_g$ such that $|r_j(a)-{r}_g(a)|\leq \epsilon |{r}_g(a)|$
 \end{defi}
 The following result follows from the definition of near-indistinguishability.
 \begin{thm}
Result \ref{thmmain2} extends to near-indistinguishable case with an approximation given by
$$
r_j(a)=r_g(a)+O(\epsilon_1)=\bar{r}(\bar{m})+O(\epsilon_1)+O(\epsilon_2^2)
$$ for the first type (non-scalable) and

$$
r_j(a)=(1\pm O(\epsilon_1))r_g(a)=(1\pm O(\epsilon_1))\bar{r}(\bar{m})+O(\epsilon_2^2)
$$
where $\epsilon_1$ comes from the scalable near-indistinguishability error (second type, scalable notion) and $\epsilon_2$ is the heterogeneity gap between action profile $a$ and $\bar{m}.$
\end{thm}
\subsubsection{Discussions}
One of the main motivations to study mean field games is the possibility to reduce the high complexity of interactive dynamical systems into a low-complexity  and easier to solve ones. However, the infinite mean field game system suggests a continuum of players may not be realistic in many cases of interests. Then, the question addressed in this paper  is to know whether the mean field game ideas can be used in the finite regime. We show that the answer is positive for  important classes of  payoff functions.

 An important statement is that if the asymptotic mean-field system can reduce the computational complexity then, the same analysis can be conducted in the finite regime. Moreover, the non-asymptotic mean field game model developed here does not require additional assumptions than the classical ones (namely A0 and A1) used in asymptotic mean field game theory.
 If the indistinguishability assumption fails but still the asymptotic mean field system is easily solvable then, one can classify the finite system too by class/and type and hence reduce into a game with less number of classes (than players) and in each class the indistinguishability property holds. We refer to such games as indistinguishability per class games.
Interestingly our approximation results extends to near-indistinguishable games as well as to indistinguishable per class games.
}
\section{Concluding remarks}
We have presented a mean field framework where the indistinguishability property can be exploited to cover not only the asymptotic regime but also the non-asymptotic regime. In other words, our approximation is suitable not only for large systems but also for a small system with few players.
The framework can be used to approximate unknown functions in heterogeneous systems, in optimization theory as well as in game theory.

This work suggests several paths for future research. First, the approach introduced
here can be used in several  applications, starting from other queueing and auctions
formats, in particular to private information models  where strategies are functions  of types. Second, more progress needs to be done by considering
a less restrictive action and belief spaces that are far from the mean of the mean field.
The smoothness condition on the objective function may not be satisfied in practice.
Finally, we would like to understand how large the deviation   of the non-asymptotic result is compared to a symmetric vector (non-alignment level).

\bibliographystyle{plain}

{\bf Hamidou Tembine} (S'06-M'10-SM'13) received his M.S. degree from Ecole Polytechnique and his Ph.D. degree from University of Avignon. His current research interests include evolutionary games, mean field stochastic games, distributed strategic learning and applications. In 2014 Tembine received the Outstanding Young Researcher Award from IEEE ComSoc.
He was the recipient of 5 best paper awards and has co-authored  two books. More details can be  found at tembine.com

\end{document}